\documentclass[twocolumn]{article}
\usepackage{amsmath,amssymb,amsthm}
\usepackage{times}
\usepackage[latin1]{inputenc}
\usepackage[T1]{fontenc}
\usepackage{latex8}
\usepackage{graphicx}
\usepackage{verbatim}
\usepackage{makeidx}
\usepackage{url}
\newtheorem{example}{Example}
\newtheorem{theorem}{Theorem}
\newtheorem{lemma}{Lemma}
\newtheorem{remark}{Remark}

\begin{document}
\title{PolyAdd: Polynomial Formal Verification of Adder Circuits}
\author{
 {\centering{\begin{tabular}{c}
Rolf Drechsler \\
\\
Institute of Computer Science \\
University of Bremen \\
28359 Bremen, Germany  \\
\multicolumn{1}{c}{drechsler@uni-bremen.de}
\end{tabular}
 }} }
\maketitle
\thispagestyle{empty}
\begin{abstract}
Only by formal verification approaches functional correctness can be ensured. While for many circuits fast verification is possible, in other cases the approaches fail. In general no efficient algorithms can be given, since the underlying verification problem is NP-complete.

In this paper we prove that for different types of adder circuits polynomial verification can be ensured based on BDDs. While it is known that the output functions for addition are polynomially bounded, we show in the following that the entire construction process can be carried out in polynomial time. This is shown for the simple Ripple Carry Adder, but also for fast adders like the Conditional Sum Adder and the Carry Look Ahead Adder. Properties about the adder function are proven and the core principle of polynomial verification is described that can also be extended to other classes of functions and circuit realizations. 
\end{abstract}
\section{Introduction}

Ensuring the functional correctness of circuits and systems is one of the major challenges in today's circuit and system design. While simulation and emulation approaches reach their limits due to the complexity of the system under verification according to Moore's Law, only formal proof techniques can ensure correctness according to the specification (see e.g.~\cite{Dre:2004,Dre:2018}). In these approaches proof engines, like BDD, SAT or SMT, are applied. 

In practice these techniques work often well and can handle circuits of several million gates. But it might also happen that the proof fails due to run time or memory constraints. One of the major difficulties is that this can hardly be predicted resulting in non-robust behavior of the tools. For this, a deeper understanding is required which circuits can be handled efficiently and for which ones the formal approach will fail. 

In the context of the highly relevant class of arithmetic circuits early studies on BDDs have shown that they are not well-suited to verify multipliers \cite{Bry:91}, but using dedicated data structures, like *BMDs \cite{BC:95} it was possible to represent the output functions of a multiplier  polynomially. In \cite{KMB+:97} it has been shown that not only the outputs can be represented, but for a specific type of Wallace tree multiplier the complete verification can be carried out polynomially. 

In this paper, we consider circuits for addition of two binary numbers. While it is well known that the BDD size for the adder function is only linear in the bit size \cite{Bry:86}, we show that the complete construction process of the BDD is also bounded polynomially. This is shown for three different adder architectures, namely the {\em Ripple Carry Adder} (RCA), the {\em Conditional Sum Adder} (CoSA) and the {\em Carry Look Ahead Adder} (CLA). Theoretical bounds on the BDD sizes are proven and it is shown that the complete symbolic simulation starting from the inputs to the outputs of the circuit can be carried out polynomially. Furthermore, for specific functions upper bounds on the BDD size are proven. 

The paper is structured as follows: In Section \ref{se:notdef} notations and definitions are reviewed to make the paper self-contained. The adder function and BDDs are introduced. For the three adders the circuit realization is reviewed in Section \ref{se:circ_real}. in Section \ref{se:poly_ver} for the three adder architectures it is proven that formal verification can be done efficiently. Finally, the results are summarized and open problems are addressed.

\section{Notation and Definition}\label{se:notdef}

Let $f: {\bf B}^n \rightarrow {\bf B}$ be a Boolean function over variable set $X_n = \{x_1, \ldots, x_n \}$. 

\begin{figure}[t]
\begin{center}
\includegraphics[scale=0.9]{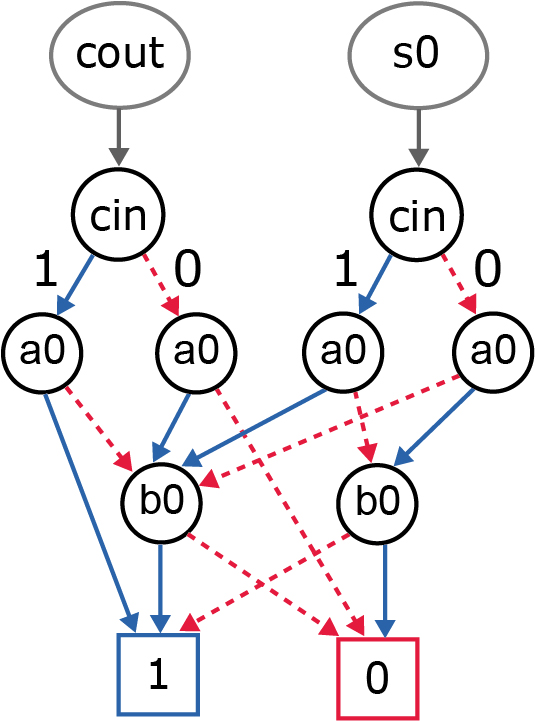}
\end{center}
\caption{BDD for full adder}\label{fi:fa_bdd}       
\end{figure}

\subsection{Adder Function}\label{sub:sym}
\vspace{-0.15cm}

Let $a$, $b$ and $s$ be three binary numbers of $n$ bits, 
where $s$ is the sum of $a$, $b$ and an incoming carry bit $c_{-1}$. The relation between the sum $s$ and the operands $a$ and $b$ can be described by the following two equations:
\begin{eqnarray}
\forall_{i=0}^{n-1} c_i = a_ib_i  + a_ic_{i-1} + b_ic_{i-1} \\ \nonumber 
 \\
\forall_{i=0}^{n-1} s_i = a_i \oplus b_i \oplus c_{i-1} \label{eq:sum}
\end{eqnarray}
The variable $c_i$ is called the $i$-th carry bit. The core cells of many adder architectures are the {\em Half Adder} (HA) and {\em Full Adder} (FA) cells realizing a 1-bit addition without or with carry input, respectively. The function table of the HA is shown in following table:
\begin{center}
\begin{tabular}{ cc|cc }
  $a_i$ & $b_i$ & $ha_1$ & $ha_0$ \\ \hline
     0     &    0    &      0     &      0    \\
     0     &    1    &      0     &      1    \\
     1     &    0    &      0     &      1    \\
     1     &    1    &      1     &      0     
 \end{tabular}
\end{center}
It is easy to see that the function $ha_1$ can be realized by an AND-gate, while $ha_0$ is described by an $\oplus$-gate, i.e.:
$$ ha_1 = a_i \cdot b_i \hspace{1cm} ha_0 = a_i \oplus b_i$$
For the FA with inputs $a_i$, $b_i$ and $c_{i-1}$ it holds:
$$ fa_1= a_i \cdot b_i + c_{i-1} \cdot (a_i + b_i) \hspace{1cm} fa_0 = a_i \oplus b_i \oplus c_{i-1}$$

\subsection{Binary Decision Diagrams}
\vspace{-0.15cm}
Reduced ordered {\em Binary Decision Diagrams} (BDDs) \cite{Bry:86,DB:98b} are {\em Directed Acyclic Graphs} (DAGs) where a Shannon decomposition
$$f=\overline{x}_if_{\overline{x}_i}+x_if_{x_i} (1 \leq i \leq n)$$ is carried out in each node. 

\begin{example}
The BDD for the FA is shown in Figure \ref{fi:fa_bdd}.
\end{example}

An important property of BDDs is that the synthesis operations, like AND, OR or composition, can be carried out in polynomial time and space. This can be described by the operator {\em if-then-else} (ite) \cite{Bry:86,BRB:90}\footnote{Notice that in the following for the discussion and the proofs BDDs without complemented edges are considered.}. A sketch of the algorithm is as follows, where {\em Rh} and {\em Rl} denote the high- and low-successors, respectively, and e.g.~{\em F1i} is the cofactor to $1$ with respect to variable $i$:

\begin{verbatim}
ite(F,G,H) {
 if (terminal case OR 
     (F,G,H) in computed-table) { 
  return result; 
} else {
 let xi be the top variable of (F,G,H);
 Rh = ite(F1i,G1i,H1i);
 Rl = ite(F0i,G0i,H0i);
 if (Rh = Rl) return Rh; 
 R = find_or_add_unique_table(v,Rl,Rh);
 insert_computed_table(F,G,H,R);
 return R;
 } 
}
\end{verbatim}
 The {\em ite}-operator has a polynomial worst case behavior, i.e.~for graphs $F$, $G$ and $H$ the result is bounded by $O(|F| \cdot |G| \cdot |H|)$. This bound holds under the assumption of an optimal hashing in $O(1)$. But also in the case of a worst case behavior of the hashing function, {\em ite} remains polynomial (see \cite{DS:2001b}).

\subsection{Symbolic Simulation}

To build the BDDs for the output signals of a circuit, the circuit is traversed in a topological order starting from the inputs. For the inputs signals the corresponding BDDs are initially generated. Then, for each gate in the circuit the corresponding synthesis operation based on {\em ite} is carried out. This process is called {\em symbolic simulation} in the following. 

\begin{example}
The symbolic simulation for a circuit consisting of a single AND gate is shown in Figure \ref{fi:symbsym}.
\end{example}

\begin{figure}[t]
\begin{center}
\includegraphics[scale=0.9]{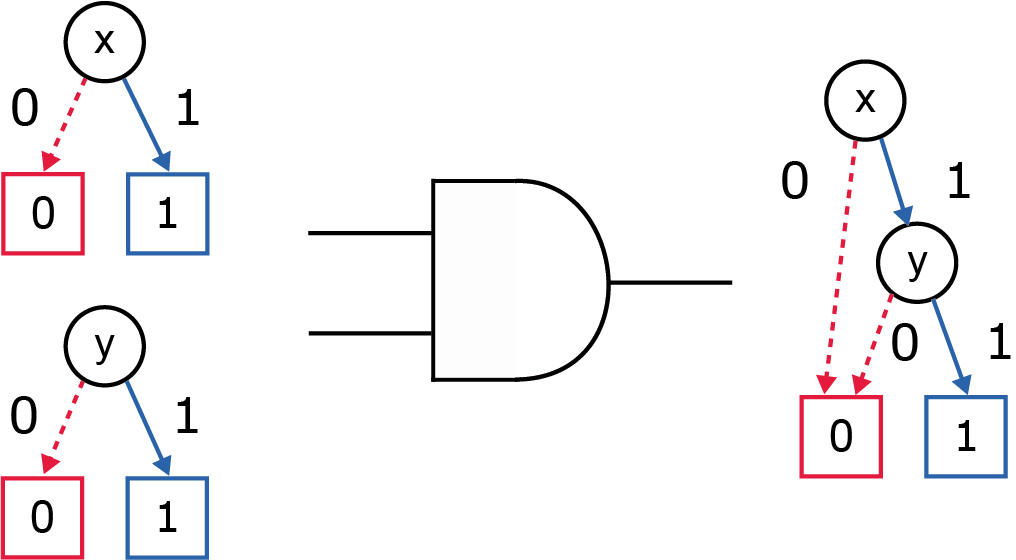}
\end{center}
\caption{Symbolic simulation for AND gate}\label{fi:symbsym}       
\end{figure}

\section{Circuit Realization}\label{se:circ_real}

In this section different realizations for adder circuits are briefly reviewed. Only the basic principles are reviewed as far as it is needed for making the paper self-contained. For more details see \cite{BDM:2005}.

\subsection{Ripple Carry Adder}

The {\em Ripple Carry Adder} (RCA) simply consists of a sequence on $n$ full adders. The cells are connected via the carry chain (see Figure \ref{fi:rca}). 
\begin{figure}[t]
\begin{center}
\includegraphics[scale=0.35]{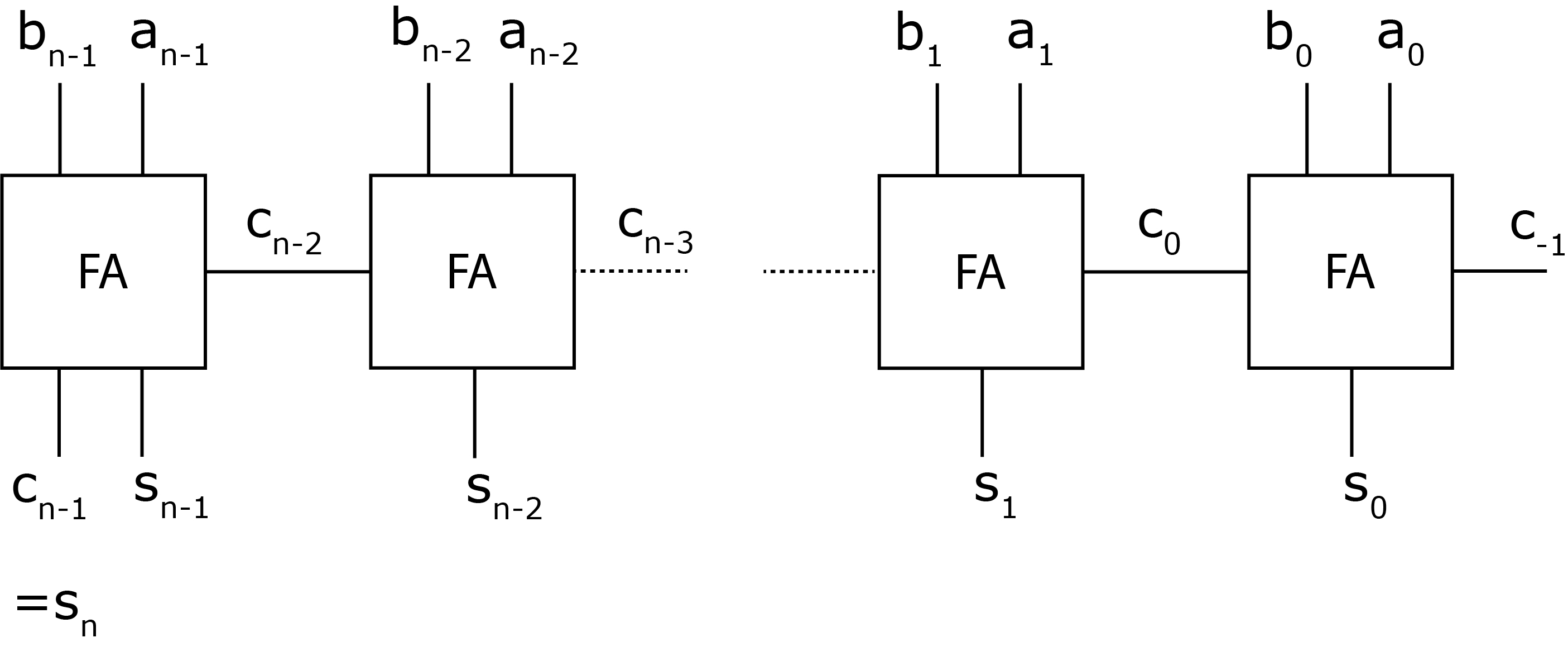}
\end{center}
\caption{Ripple Carry Adder}\label{fi:rca}       
\end{figure}

The RCA is very area efficient, since it only requires a linear number of gates. But the RCA is also very slow, since the delay -- measured in the number of gates that has to be traversed -- is also linear in the number of inputs. 

\subsection{Conditional Sum Adder}

The {\em Conditional Sum Adder} (CoSA) can be recursively described. While the lower $n/2$ bits are computed by a CoSA of bit-width $n/2$, for the higher $n/2$ bits the result is computed by two CoSAs in parallel, where one assumes an incoming carry, while the other does not. Thus, the adder makes use of the fact that the higher bits only depend on the incoming carry from the lower half. Both results are pre-computed and the correct result is selected by a multiplexer stage. The computation scheme is shown in Figure \ref{fi:cosa}. For the 1-bit adders, simply full adders can be used. 

\begin{figure}[t]
\begin{center}
\includegraphics[scale=0.375]{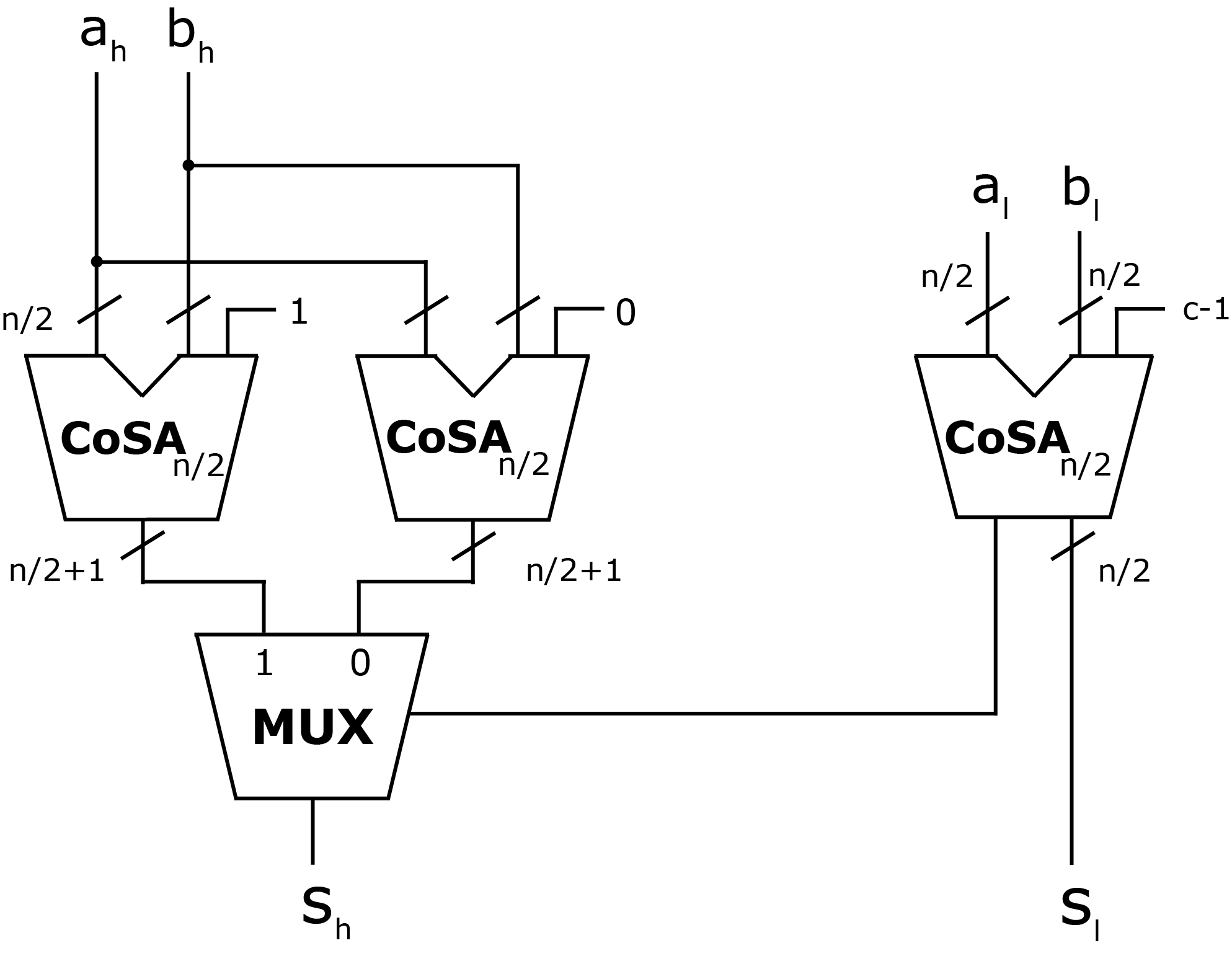}
\end{center}
\caption{Conditional Sum Adder}\label{fi:cosa}       
\end{figure}

The CoSA is a fast adder, i.e.~it has a depth of $O(log(n))$. The circuit has a gate count of $O(n \cdot log(n))$.

\subsection{Carry Look Ahead Adder}

The {\em Carry Look Ahead Adder} (CLA) makes use of a fast prefix computation in a block $P_n$ (see Figure \ref{fi:cla}). From Equation (\ref{eq:sum}) it is obvious that it is sufficient to compute the carry bits $c_i$ for all $i$. This can be done based on parallel prefix computation of the generation and propagation properties for addition. These are described by function $g$ and $p$, respectively:
\begin{enumerate}
\item For $0 \leq i <n$: $p_{i,i}= a_i \oplus b_i$, $g_{i,i} = a_i \cdot b_i$
\item For $i \leq k < j$: $p_{j,i} = p_{k,i} \cdot p_{j,k+1}$, \\ $g_{j,i} = g_{j, k+1} + (g_{k,i} \cdot p_{j,k+1})$, 
\end{enumerate}
This means that either a carry bit is generated in the upper part or a carry is generated in the lower part and is propagated through the higher part. 
Thus, the carry bits can be computed as ($0 \leq i < n$):
$$ c_i = g_{i,0} + p_{i,0} \cdot c_{-1} $$

\begin{figure}[t]
\begin{center}
\includegraphics[scale=0.375]{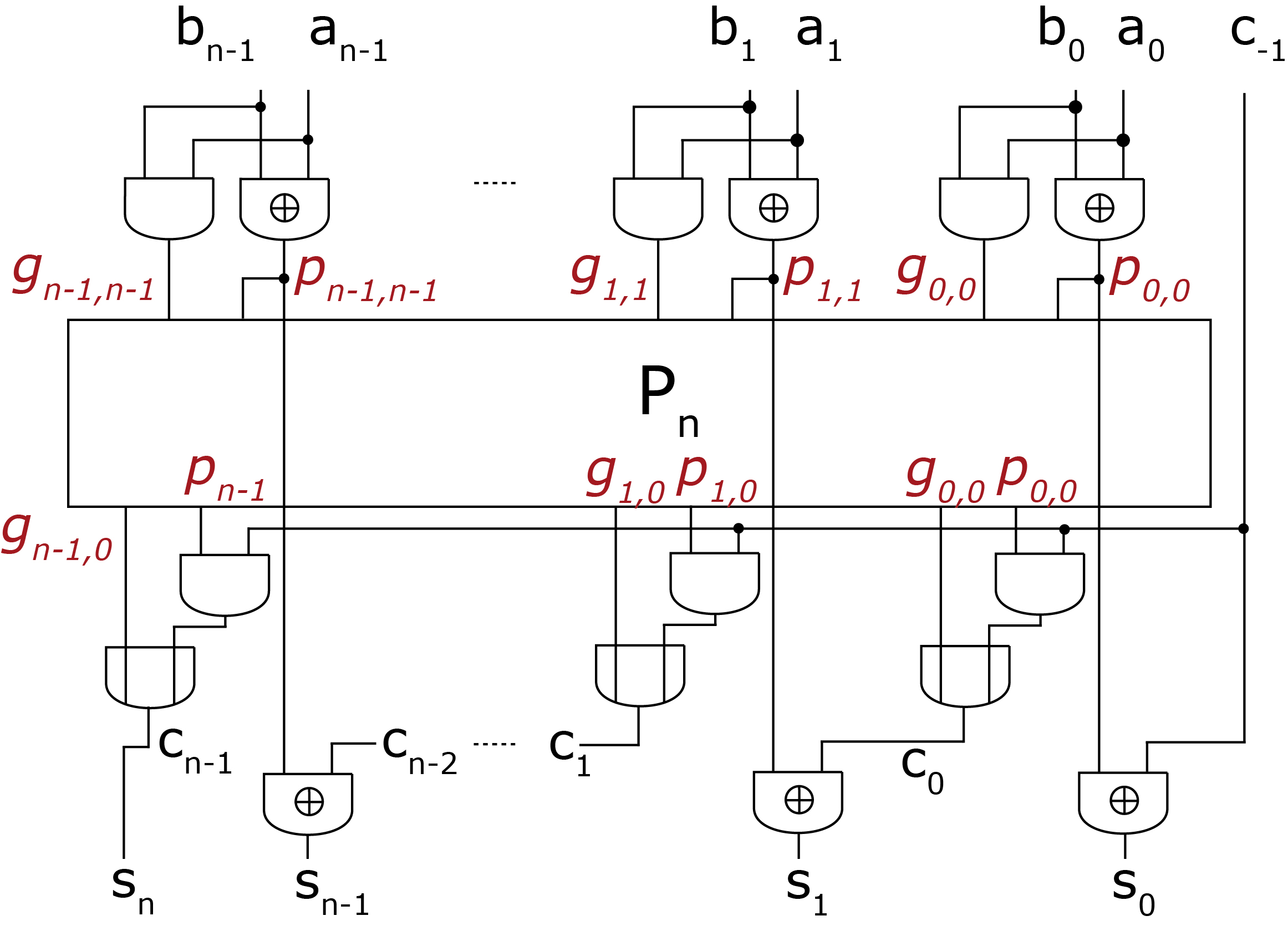}
\end{center}
\caption{Carry Look Ahead Adder}\label{fi:cla}       
\end{figure}

The CLA has a logarithmic depth and a size linear in the number of input variables. 

\begin{figure*}[t]
\begin{center}
\includegraphics[scale=0.6]{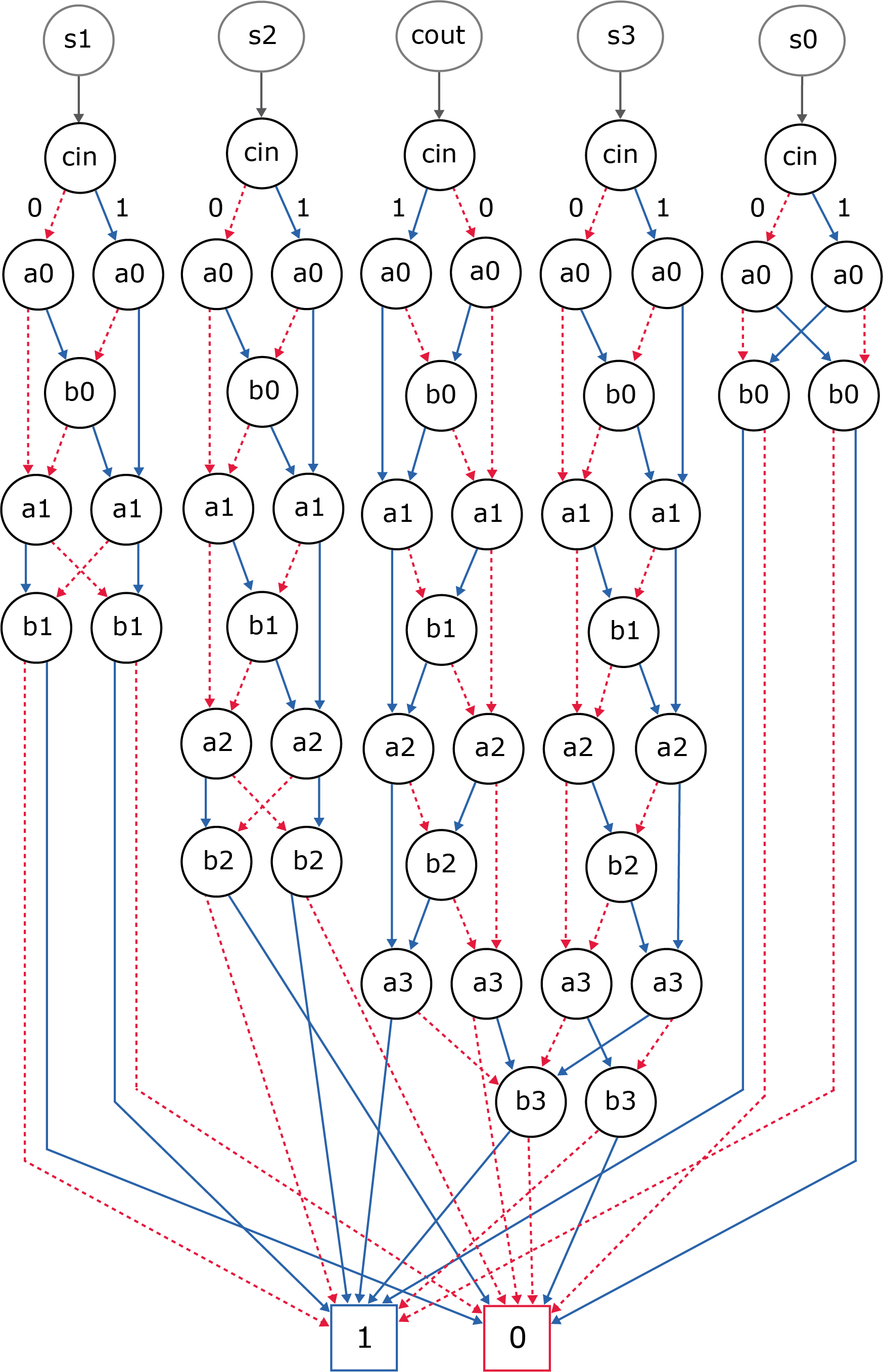}
\end{center}
\caption{BDD for 4-bit adder function}\label{fi:bdd_add}       
\end{figure*}

\section{Polynomial Verification}\label{se:poly_ver}

It is well known that the size of BDDs for the adder function is dependent on the variable ordering. It has also been proven that the BDD size is linearly bounded (see Section 4.4 in \cite{Weg:2000}), where exact estimates are given for BDD sizes. There, addition without the incoming carry bit has been considered. The results can be extended to also consider the incoming carry bit as it is required for all adder circuits in the following.
\begin{theorem}\label{th:adder_bdd}
\begin{enumerate}
\item The sum bit $s_i$ of an adder has the BDD size bounded by $3 \cdot i+7$.
\item The carry bit $c_i$ of an adder has the BDD size bounded by $3 \cdot i+6$.
\end{enumerate}
\end{theorem}
\begin{proof}
We use the interleaved variable ordering from the least to the most significant bits. 
For the sum bits the results from Lemma 4.4.2 in \cite{Weg:2000} can be generalized, where an upper bound of $3 \cdot i+5$ has been proven for the adder function without an incoming carry bit. For the additional carry bit two more nodes are required, i.e.~one for the carry bit itself and one for the $a_0$ variable. 

The same argument holds for the carry bit, but here on the lowest level one node is saved, since in case of generation by $a_i$ and the incoming carry, $b_i$ does not have to be tested any more (see Figure  \ref{fi:bdd_add} for the case of 4 variables). 
\end{proof}

It is important to notice that these results were always related to the representation size of the output functions, but not for the entire construction process. 
\begin{remark}\label{re:no_exact_bounds}
In the following, detailed bounds are not provided, since the goal of this paper is to show that the construction process is polynomial. 
\end{remark}
Thus, it is sufficient to show that each individual step can be carried out in polynomial time and space. We make use of the following observation: 
\begin{remark}
If for each internal signal the size of the BDD representation and the number of gates in the circuit is polynomially bounded in the number of inputs $n$, the whole circuit can be formally verified in polynomial time due to the polynomially bounded synthesis operations on BDDs. 
\end{remark}
This method can be applied to general circuits, but is used for adders only in the following. For the adder circuits from Section \ref{se:circ_real} the upper bounds hold, that each circuit only has a number of gates polynomial in the number of inputs $n$.

\subsection{Ripple Carry Adder}\label{sub:rca}

For the RCA it is very simple to see that the complete construction is polynomially bounded. For the HA of the least significant bit and all FAs the BDD can be locally constructed and has only a constant size. Due to the structure of the RCA each carry output of a cell is connected to the carry input of the next cell. The substitution of the input variable can be carried out by the compose algorithm based on {\em ite} and has a polynomial worst-case complexity. Furthermore, according to Theorem \ref{th:adder_bdd} the size of the BDD for the carry signal for all $i$ is always linear. Thus, the whole construction process is polynomially bounded, since the composition only has to be carried out $n$ times. 
\begin{theorem}\label{th:rca}
The BDD for the RCA can be constructed polynomially.  
\end{theorem}

\subsection{Conditional Sum Adder}\label{sub:cosa}

The $n$ bit CoSA consists of three CoSAs of bit-size $n/2$ and a multiplexer stage. From Theorem \ref{th:adder_bdd} it follows that each of the connecting signals shown in Figure \ref{fi:cosa} can be represented by a BDD of linear size. Only the carry inputs have to be set to $0$ and $1$, respectively. The only operation that has to be carried out is the one corresponding to the MUX unit. But this can be described by {\em ite} and is polynomially bounded. Thus, we obtain: 
\begin{theorem}\label{th:cosa}
The BDD for the CoSA can be constructed polynomially.  
\end{theorem}

\begin{remark}
The results of Theorems \ref{th:rca} and \ref{th:cosa} can easily be generalized to further adder types that are based on full adders connected together using MUX cells, like e.g.~the {\em Carry Select Adder} in \cite{BDKR:96} with a run time of $O(\sqrt{n})$.
\end{remark}

\subsection{Carry Look Ahead Adder}

In the CLA the sum bits are computed by determining the carry bits first and finally EXOR-ing them with the corresponding $a_i$ and $b_i$ inputs according to Equation (\ref{eq:sum}). Thus, the core circuit computes the carry bits starting based on the property of generation and propagation, i.e.~functions $p$ and $g$. The union of propagation intervals is based on Boolean AND-operations, i.e.~larger interval only propagates a carry bit, if the left and the right part of the interval do so. For the generation part it holds that either the left part (using the higher bits) already propagates or the lower part generates, while the higher part propagates. In both cases, the structure consists of AND- and OR-operations only and it can be seen that the whole structure can be represented by BDDs of polynomial size. More formally, this can be proven as follows: 
\begin{lemma}
\begin{enumerate}
\item Function $p_{j,i}$ has the BDD size bounded by $3 \cdot (j-i+1)$ ($j>i$).
\item Function $g_{j,i}$ has the BDD size bounded by $3 \cdot (j-i) + 2$ ($j>i$).
\end{enumerate}
\end{lemma}
\begin{proof}
For function $p_{j,i}$ it holds: 
$$p_{j,i} = (a_j \oplus b_j)(a_{j-1} \oplus b_{j-1}) \ldots (a_i \oplus b_i)$$
The BDD for the EXOR of two variables has three nodes. Since each variable only appears once, the corresponding BDDs can simply be connected (see Figure \ref{fi:bdd_p} for the case of 4 variables).  

Since the BDD is a cannonical representation, in 
$$g_{j,i} = g_{j, k+1} + (g_{k,i} \cdot p_{j,k+1})$$ the choice of $k$ does not influence the BDD size and we choose $k=j-1$ resulting in $$g_{j,i} = g_{j, j} + (g_{j-1,i} \cdot p_{j,j}).$$ For each pair of variables  $a_l$, $b_l$ at most 3 nodes can be generated (see Figure \ref{fi:bdd_g} for the case of 4 variables). For the top variables even one more node is saved. 
\end{proof}
\begin{figure}[t]
\begin{center}
\includegraphics[scale=0.64]{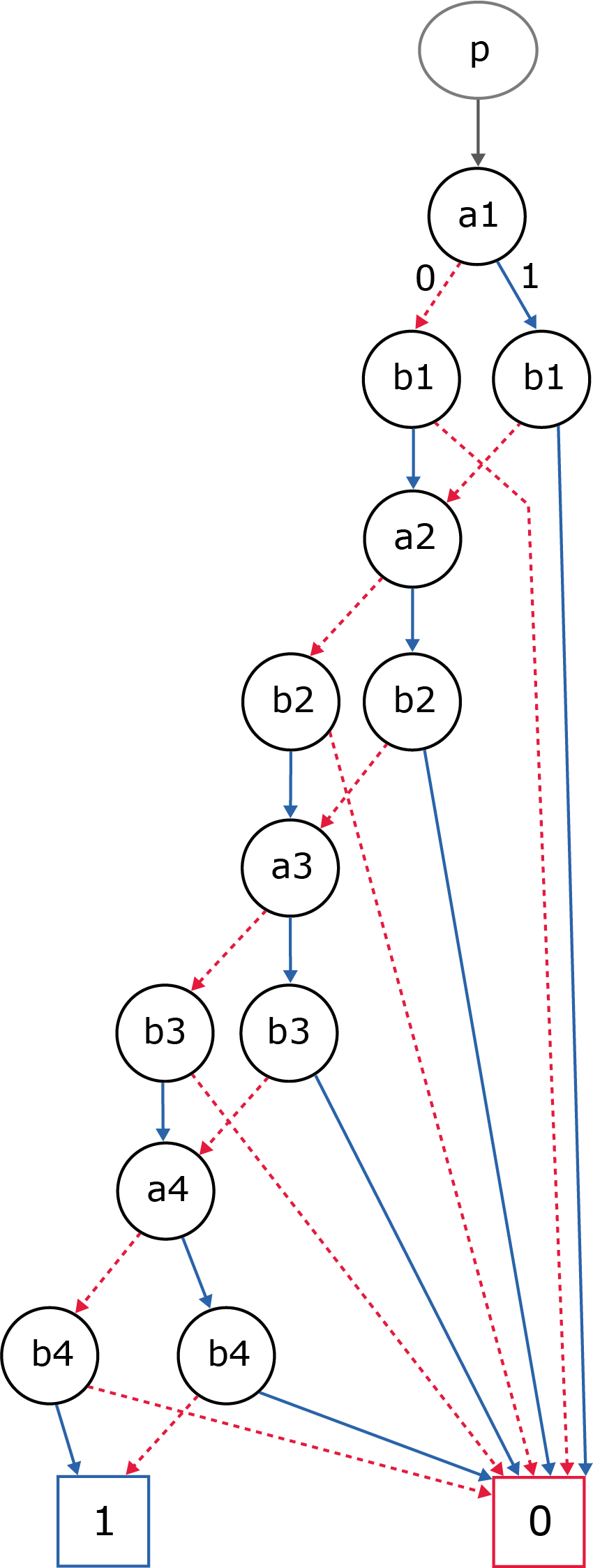}
\end{center}
\caption{BDD for $p$ function for 4 variables}\label{fi:bdd_p}       
\end{figure}
\begin{figure}[t]
\begin{center}
\includegraphics[scale=0.64]{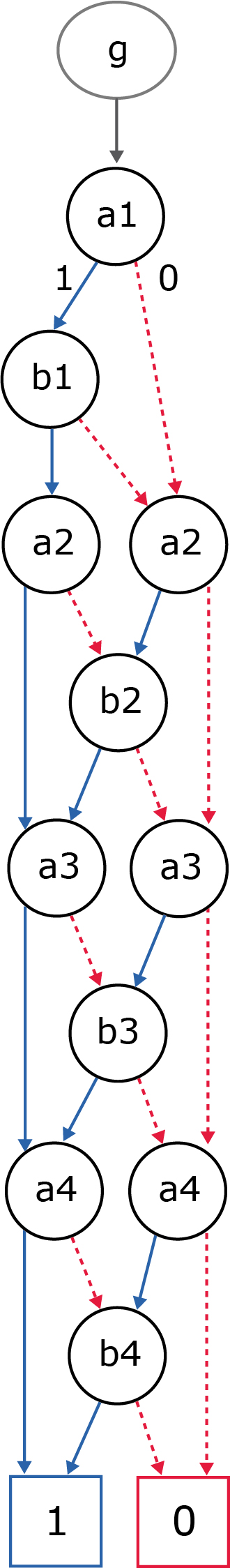}
\end{center}
\caption{BDD for $g$ function for 4 variables}\label{fi:bdd_g}       
\end{figure}
Based on this observation, the whole BDD for the CLA can be computed based on {\em ite}. 

\begin{theorem}
The BDD for the CLA can be constructed polynomially.  
\end{theorem}

\section{Conclusion}\label{sec:concl}

In this paper it has been proven for three different adder architectures that the complete formal verification process can be carried out polynomially. It was proven that the underlying BDDs remain polynomial during the whole construction process. This was ensured by proving upper bounds on the BDD sizes for each internal signal. While the BDD sizes for the outputs of the adder functions were known to be polynomially bounded, this is the first time that for efficient adder circuits of logarithmic run time a polynomial proof process could be ensured. 

It is focus of future work to identify further classes of circuits and functions that can be polynomially verified using BDDs. Furthermore, alternative proof engines on the Boolean level, like SAT or O(K)FDDs, can be considered. Also extension to the word-level, like SMT or WLDDs, will be studied.  

\section*{Acknowledment}
Parts of this work have been supported by DFG within the Reinhart Koselleck Project {\em PolyVer: Polynomial Verification  of Electronic Circuits} (DR 287/36-1). Furthermore, the author likes to thank Alireza Mahzoon for helpful comments and discussions. 
\bibliographystyle{IEEEtran}
\bibliography{lit_bank,fey_loc,grosse_loc,lit_bank_ext}

\begin{thebibliography}{10}
\providecommand{\url}[1]{#1}
\csname url@samestyle\endcsname
\providecommand{\newblock}{\relax}
\providecommand{\bibinfo}[2]{#2}
\providecommand{\BIBentrySTDinterwordspacing}{\spaceskip=0pt\relax}
\providecommand{\BIBentryALTinterwordstretchfactor}{4}
\providecommand{\BIBentryALTinterwordspacing}{\spaceskip=\fontdimen2\font plus
\BIBentryALTinterwordstretchfactor\fontdimen3\font minus
  \fontdimen4\font\relax}
\providecommand{\BIBforeignlanguage}[2]{{%
\expandafter\ifx\csname l@#1\endcsname\relax
\typeout{** WARNING: IEEEtran.bst: No hyphenation pattern has been}%
\typeout{** loaded for the language `#1'. Using the pattern for}%
\typeout{** the default language instead.}%
\else
\language=\csname l@#1\endcsname
\fi
#2}}
\providecommand{\BIBdecl}{\relax}
\BIBdecl

\bibitem{Dre:2004}
R.~Drechsler, \emph{\BIBforeignlanguage{USenglish}{Advanced Formal
  Verification}}.\hskip 1em plus 0.5em minus 0.4em\relax Kluwer Academic
  Publishers, 2004.

\bibitem{Dre:2018}
------, \emph{\BIBforeignlanguage{USenglish}{Formal System
  Verification}}.\hskip 1em plus 0.5em minus 0.4em\relax Springer, 2018.

\bibitem{Bry:91}
R.~Bryant, ``\BIBforeignlanguage{USenglish}{On the complexity of {VLSI}
  implementations and graph representations of {B}oolean functions with
  application to integer multiplication},''
  \emph{\BIBforeignlanguage{USenglish}{IEEE Trans. on Comp.}}, vol.~40, pp.
  205--213, 1991.

\bibitem{BC:95}
R.~Bryant and Y.-A. Chen, ``\BIBforeignlanguage{USenglish}{Verification of
  arithmetic functions with binary moment diagrams},'' in
  \emph{\BIBforeignlanguage{USenglish}{Design Automation Conf.}}, 1995, pp.
  535--541.

\bibitem{KMB+:97}
\BIBentryALTinterwordspacing
M.~Keim, M.~Martin, B.~Becker, R.~Drechsler, and P.~Molitor,
  ``\BIBforeignlanguage{USenglish}{Polynomial formal verification of
  multipliers},'' in \emph{\BIBforeignlanguage{USenglish}{VLSI Test Symp.}},
  1997, pp. 150--155. [Online]. Available:
  \url{http://ira.informatik.uni-freiburg.de/papers/Year_97/KMBDM_97.ps.gz}
\BIBentrySTDinterwordspacing

\bibitem{Bry:86}
R.~Bryant, ``\BIBforeignlanguage{USenglish}{Graph-based algorithms for
  {B}oolean function manipulation},'' \emph{\BIBforeignlanguage{USenglish}{IEEE
  Trans. on Comp.}}, vol.~35, no.~8, pp. 677--691, 1986.

\bibitem{DB:98b}
R.~Drechsler and B.~Becker, \emph{\BIBforeignlanguage{USenglish}{Binary
  Decision Diagrams -- Theory and Implementation}}.\hskip 1em plus 0.5em minus
  0.4em\relax Kluwer Academic Publishers, 1998.

\bibitem{BRB:90}
K.~Brace, R.~Rudell, and R.~Bryant, ``\BIBforeignlanguage{USenglish}{Efficient
  implementation of a {BDD} package},'' in
  \emph{\BIBforeignlanguage{USenglish}{Design Automation Conf.}}, 1990, pp.
  40--45.

\bibitem{DS:2001b}
R.~Drechsler and D.~Sieling, ``\BIBforeignlanguage{USenglish}{Binary decision
  diagrams in theory and practice},''
  \emph{\BIBforeignlanguage{USenglish}{Software Tools for Technology
  Transfer}}, vol.~3, pp. 112--136, 2001.

\bibitem{BDM:2005}
B.~Becker, R.~Drechsler, and P.~Molitor,
  \emph{\BIBforeignlanguage{USenglish}{{Technische Informatik - Eine
  Einf\"uhrung}}}.\hskip 1em plus 0.5em minus 0.4em\relax Pearson Studium,
  2005.

\bibitem{Weg:2000}
I.~Wegener, \emph{\BIBforeignlanguage{USenglish}{Branching Programs and Binary
  Decision Diagrams - Theory and Application}}.\hskip 1em plus 0.5em minus
  0.4em\relax SIAM Monographs on Discrete Mathematics and Applications, 2000.

\bibitem{BDKR:96}
B.~Becker, R.~Drechsler, R.~Krieger, and S.~Reddy,
  ``\BIBforeignlanguage{USenglish}{A fast optimal robust path-delay-fault
  testable adder},'' in \emph{\BIBforeignlanguage{USenglish}{European Design \&
  Test Conf.}}, 1996, pp. 491--498.

\end{thebibliography}

\end{document}